\documentclass{article}
\usepackage{amsmath}
	\PassOptionsToPackage{numbers, compress}{natbib}

\usepackage[preprint]{neurips_2021}

\usepackage[utf8]{inputenc} 
\usepackage[T1]{fontenc}    
\usepackage{hyperref}       
\usepackage{url}            
\usepackage{booktabs}       
\usepackage{amsfonts}       
\usepackage{nicefrac}       
\usepackage{microtype}      
\usepackage{xcolor}         

\usepackage{amsthm}
\usepackage{cases}
\usepackage{bm}
\usepackage{algorithm}
\usepackage{algorithmic}
\usepackage{cases}
\usepackage{subfigure}
\usepackage{multirow}
\usepackage{color} 
\usepackage{enumerate}
\usepackage{etoolbox}
\usepackage{mathrsfs}
\newcommand{\MyMapTemplatePrefix}[4]{\expandafter#1\csname#3#4\endcsname{#2{#4}}}
\newcommand{\MyMapTemplatePrefixNew}[5]{\expandafter#1\csname#4#5\endcsname{#2{#3{#5}}}}
\forcsvlist{\MyMapTemplatePrefix {\def} {\mathbf} {}} {A,B,C,D,E,F,G,H,I,J,K,L,M,N,O,P,Q,R,S,T,U,V,W,X,Y,Z}
\forcsvlist{\MyMapTemplatePrefix {\def} {\mathbf} {}} {a,b,c,d,e,f,g,h,i,j,k,l,m,n,o,p,q,r,s,t,u,v,w,x,y,z,1,0}
\forcsvlist{\MyMapTemplatePrefix {\def} {\widetilde} {wt}} {A,B,C,D,E,F,G,H,I,J,K,L,M,N,O,P,Q,R,S,T,U,V,W,X,Y,Z}
\forcsvlist{\MyMapTemplatePrefix {\def} {\widetilde} {wt}} {a,b,c,d,e,f,g,h,i,j,k,l,m,n,o,p,q,r,s,t,u,v,w,x,y,z} 
\forcsvlist{\MyMapTemplatePrefixNew {\def} {\widetilde}{\mathbf} {tb}} {A,B,C,D,E,F,G,H,I,J,K,L,M,N,O,P,Q,R,S,T,U,V,W,X,Y,Z}
\forcsvlist{\MyMapTemplatePrefixNew {\def} {\widetilde}{\mathbf} {tb}} {a,b,c,d,e,f,g,h,i,j,k,l,m,n,o,p,q,r,s,t,u,v,w,x,y,z}
\forcsvlist{\MyMapTemplatePrefix {\def} {\widehat} {wh}} {A,B,C,D,E,F,G,H,I,J,K,L,M,N,O,P,Q,R,S,T,U,V,W,X,Y,Z}
\forcsvlist{\MyMapTemplatePrefix {\def} {\widehat} {wh}} {a,b,c,d,e,f,g,h,i,j,k,l,m,n,o,p,q,r,s,t,u,v,w,x,y,z}
\forcsvlist{\MyMapTemplatePrefixNew {\def} {\widehat}{\mathbf} {hb}} {A,B,C,D,E,F,G,H,I,J,K,L,M,N,O,P,Q,R,S,T,U,V,W,X,Y,Z}
\forcsvlist{\MyMapTemplatePrefixNew {\def} {\widehat}{\mathbf} {hb}} {a,b,c,d,e,f,g,h,i,j,k,l,m,n,o,p,q,r,s,t,u,v,w,x,y,z}
\forcsvlist{\MyMapTemplatePrefixNew {\def} {\overline}{\mathbf} {lb}} {A,B,C,D,E,F,G,H,I,J,K,L,M,N,O,P,Q,R,S,T,U,V,W,X,Y,Z}
\forcsvlist{\MyMapTemplatePrefixNew {\def} {\overline}{\mathbf} {lb}} {a,b,c,d,e,f,g,h,i,j,k,l,m,n,o,p,q,r,s,t,u,v,w,x,y,z}
\forcsvlist{\MyMapTemplatePrefix {\def} {\mathcal}{mc}} {A,B,C,D,E,F,G,H,I,J,K,L,M,N,O,P,Q,R,S,T,U,V,W,X,Y,Z}
\forcsvlist{\MyMapTemplatePrefix {\def} {\mathbb} {mb}} {A,B,C,D,E,F,G,H,I,J,K,L,M,N,O,P,Q,R,S,T,U,V,W,X,Y,Z}

\newtheorem{thm}{Theorem}[section]
\newtheorem{defn}{Definition}  

\title{Going Deeper in Frequency Convolutional Neural Network: A Theoretical Perspective}

\author{%
   Xiaohan Zhu\\
   Nanjing University of \\
   Science and Technology\\
   \texttt{zxhan784@njust.edu.cn} \\
   \And
   Zhen Cui \\
   Nanjing University of \\
   Science and Technology\\
   \texttt{zhen.cui@njust.edu.cn} \\
   \AND
   Tong Zhang \\
   Nanjing University of \\
   Science and Technology\\
   \texttt{tong.zhang@njust.edu.cn} \\
   \And
   Yong Li \\
   Nanjing University of \\
   Science and Technology\\
   \texttt{yong.li@njust.edu.cn} \\
   \And
   Jian Yang \\
   Nanjing University of \\
   Science and Technology\\
   \texttt{csjyang@mail.njust.edu.cn} \\
}

\begin{document}

\maketitle
\begin{abstract}	
	Convolutional neural network (CNN) is one of the most widely-used successful architectures in the era of deep learning. However, the high-computational cost of CNN still hampers more universal uses to light devices. Fortunately, the Fourier transform on convolution gives an elegant and promising solution to dramatically reduce the computation cost. Recently, some studies devote to such a challenging problem and pursue the complete frequency computation without any switching between spatial domain and frequent domain. In this work, we revisit the Fourier transform theory to derive feed-forward and back-propagation frequency operations of typical network modules such as convolution, activation and pooling. Due to the calculation limitation of complex numbers on most computation tools, we especially extend the Fourier transform to the Laplace transform for CNN, which can run in the real domain with more relaxed constraints. This work more focus on a theoretical extension and discussion about frequency CNN, and lay some theoretical ground for real application.
\end{abstract}

\section{Introduction}
In the decade, numerous convolutional neural networks (CNNs) are proposed and applied to various computer vision tasks such as image classification and face recognition \cite{ref2, ref5}. 
One of the primary challenges of CNNs is the high-expensive computation cost in the training and inference stages. In particular, the efficient implementation of convolutional kernels has been a key ingredient of popular CNNs at scale. To speed up, most CNNs need the support of high-efficient GPU calculation, which hampers more universal uses to those light resource-limited devices.

The computation burden of CNNs is largely dominated by convolution layers, as illustrated in \cite{ref25, ref26}. To reduce the computation time of the convolutional operations, a number of studies have explored various efficient computation models \cite{ref10, ref13}. A promising approach for fast training and inference is to exploit the duality between spatial /frequency domain computation by the Fourier transforms. Due to its efficiency and the potential for amortization of cost, the Fourier transform has long been viewed as a natural alternative for fast convolution \cite{ref22}, e.g., Mathieu et al performed convolution in the spectral domain and recover the feature map with the inverse Fourier transform after each convolutional layer. However, such approaches merely replaced the operations inside the convolution layer \cite{ref6} or pooling layer \cite{ref7}, they usually require computationally-intensive Fourier Transforms (FT) and inverse FT at the boundary of every layer.

Although FT accelerates the computation of convolution unit, the entire network still performs part of the operations in the spatial domain. In summary, two challenges exist in previous work: 
1) they are not capable of thorough frequency transform for the neural network, e.g., some critical components such as non-linear activation function (ReLU, Sigmoid) still perform in the spatial domain, and the gradient back-propagation is difficult in the frequency domain; 
2) they exploit the complex-coefficient spectral parameterization of the convolutional filters, which is not programming friendly.

To cope with this issue, in this paper, we theoretically study the frequency operation of convolution neural network. First, we revisit the Fourier transform theory to derive feed-forward and back-propagation frequency operations of basic network components such as convolutional computation as well as pooling, Sigmoid, ReLU and fully connected computation, which enables an end-to-end frequency network. The complex number space in the Fourier Transform makes it difficult to be deployed onto the most popular computation tools. To address this problem, we extend the Fourier transform to the Laplace transform for convolution neural network. In theory, the Laplace frequency operations can be completely performed in the real space, and the satisfactory condition of Laplace transform is more relaxed than the Fourier transform. In summary, we attempt to build the deeper theoretical fundament for the complete frequency convolution neural network in the complex or real number domain. At the same time, we expect the theoretical extension and discussion about frequency CNN can benefit the essential solution to accelerate the training and inference of convolution neural network. 

\section{Related work}
\label{RW}
There are multiple works based on the Fourier transform and Laplace transform, together with their variants. Specifically, there have been a part of works which attempt to construct neural network units in frequency domain. Below, we introduce the Fourier transformation together with its variants, as well as some previous works which attempted to accelerate the inference of deep neural networks by designing frequency-domain deep networks.

FT is a linear integral transformation used for signal in the spatial/frequency domain between the transformation, and there are various variations in different applications \cite{ref3}. Convolution in spatial domain is equivalent to point-wise multiplication in frequency domain, which has been known for a long time. And the relevant research and application may simply propose the product substitution convolution operation in the frequency domain yet, rather than transforming the entire convolutional neural networks (CNNs) into the frequency domain. \cite{ref6} shows the feasibility of using simpler point-wise multiplication to replace convolution in the frequency domain, but this method only replaces the operation inside the convolution layer, requiring computationally intensive fast Fourier transform(FFT) and inverse FFT at the boundary of each layer; 

\cite{ref10} proposed Fourier Convolutional Neural Network(FCNN) to speed up the training time without reducing the effectiveness, but only involves convolution and pooling of frequency domain, where the pooling is realized by frequency truncation which is similar to the method in \cite{ref7}; By parameterizing the integral kernel directly in the Fourier space~\cite{ref19}, a new type of neuron operator is developed, which enables an efficient and expressive architecture for faster and more accurate solution of partial differential equations. \cite{ref20} designed Deep Fourier Channel Attention Network (DFCAN) to solve the problem of image super-resolution in optical microscopy. 
Inspired by this work, we propose to map the whole operation of CNNs into the frequency domain in theory.

Based on FFT, \cite{ref7} proposed the spectrum pooling method, which performs dimension reduction by truncating the representation in the frequency domain, and further promotes the realization of CNNs in the frequency domain. The work in \cite{ref22} further extends the concept of spectrum pooling by mapping the operation and parameters of the entire convolution into the frequency domain. Based on discrete cosine transform(DCT), \cite{ref9} conducted spectral convolution for getting more diverse feature of information.By extracting features in spatial and frequency domain for fusion, frequency domain learning can also be used as a prior for network compression and pruning \cite{ref13, ref14, ref15}. FreshNets are developed by \cite{ref16}, which use DCT to convert filter weights to the frequency domain. 
Through the DCT, the channel attention mechanism was reconsidered by \cite{ref21} for frequency domain analysis. It was proved mathematically that the traditional global average pooling was a special case of feature decomposition in the frequency domain.

Different from all the works above, in this paper, we attempt to build an end-to-end neural network in frequency domain based on the Fourier transform and Laplace transform.

%

\section{Spectral representation}\label{sec:spectral}

Transforming convolutional network to frequency domain can not only reduce high-computational burden, but also compact model parameters through flexible spectrum selection. The most classic calculation is the Fourier transform. Given an input function $f(x)$, the existence of the Fourier transform is conditioned by absolutely integrable property, i.e., $\int_{-\infty}^\infty |f(x)|<\infty$, and the piece-wise smoothness of $f(x)$. To convert full CNN into frequency domain, below we provide the detailed theoretical derivation and dissection in the functional space for those basic network modules.

\subsection{Convolution}
In the spatial domain, given two functions $f_{1}$ and $f_{2}$ on the one-dimension signal $x$, the convolution can be defined as follows.

\begin{defn}\label{dfn:1d_conv}
	Given two piece-wise smooth and absolutely integrable functions $f_{1}(x)$ 
	and $f_{2}(x)$ defined in the interval of $ (-\infty ,\infty) $,
	their convolution is defined as:
	\begin{equation}
	f_{1}( x)\ast f_{2}(x)=\int_{- \infty}^{\infty}f_{1}(\xi)f_{2}( x - \xi )d\xi.\label{eqn:1d_conv}
	\end{equation}
\end{defn}
in which * means convolution operation. 

Based on the definition above, here, we derive the spectral representation of the convolution in frequency domain. First, let $f(x)$ represent a piece-wise smooth and absolutely integrable function on the signal $x$, its one-dimensional Fourier transform can be written as:
\begin{equation}\label{eqn:Fconv}
F(\omega )=\int_{-\infty}^{\infty}f(x)e^{-i \omega x}dx,
\end{equation}
\begin{equation}\label{eqn:reFconv}
f(x)=\frac{1}{2 \pi}\int_{-\infty}^{\infty}F(\omega )e^{i \omega x}d\omega,
\end{equation}
where $F( \omega )$ is the frequency signal after Fourier transform, $i$ is the imaginary unit, $ \omega $ is the frequency spectrum, and $f(x)\leftrightarrow F(w)$ represents the transformation and inverse transformation between the signal and the spectrum.


Then, based on Eqn.~(\ref{eqn:Fconv}),~(\ref{eqn:reFconv}) and Definition \ref{dfn:1d_conv}, we derive the following Theorem~\ref{thm:conv}, where the detailed proof can be found in the supplemental material (Appendix A.1). 
\begin{thm} \label{thm:conv}
	The convolution of two functions is equivalent to taking the dot product in the frequency domain, formally,
	\begin{equation}
	f_{1}(x)\ast f_{2}(x)\leftrightarrow F_{1}( \omega )F_{2}( \omega ),
	\end{equation}
	in which $ F_{1}( \omega ) $ and $ F_{2}( \omega ) $ are the Fourier transforms of $ f_{1}( x ) $ and $ f_{2}( x ) $, respectively. And $ \leftrightarrow $ is the equivalent representations of transformation and its inverse.
\end{thm}
According to Theorem~\ref{thm:conv}, the Fourier transformation of convolution in spatial domain is equivalent to the point-wise product in the frequency domain. Therefore, we can derive the feed-forward frequency operation of convolution by first transforming both the input signal and convolutional filter into frequency domain, and then conduct element-wise multiplication.  

To enable the end-to-end frequency neural network with thorough frequency transform, we derive the back-propagation frequency convolution based on the Theorem of derivative about convolution. It's can be found in Appendix A.2. 
Furthermore, for the derivative of higher order, the corresponding back-propagation in frequency domain can be obtained by simply multiplying the higher-order power of $i\omega$.

We further extend the convolution in frequency domain from one-dimensional signal processing to multidimensional signal transformation. Here, we take the 2-dimensional input as the representative, and $ f_{1}(x,y) $ and $ f_{2}(x,y) $ are the function of two-dimension inputs. First, the definition of 2-dimensional convolution in spatial domain is as follows.
\begin{defn} \label{dfn:2d_conv}
	The 2-dimensional function which defined in the interval of $ ( -\infty ,\infty ) $, their convolution is defined as:
	\begin{equation}
	f_{1}(x,y)\ast f_{2}(x,y)=\int_{-\infty}^{\infty}\int_{-\infty}^{\infty}f_{1}(\xi _{x},\xi _{y})f_{2}(x- \xi _{x},y- \xi _{y}).
	\end{equation} \label{eqn:2d_conv}
\end{defn}

Then, the Fourier and inverse Fourier transform can be written as follows:
\begin{equation}
F(u, v)=\int_{-\infty}^{\infty}\int_{-\infty}^{\infty}f(x, y)e^{-i2\pi (ux+vy)}dxdy,
\end{equation}
\begin{equation}
f(x, y)=\int_{-\infty}^{\infty}\int_{-\infty}^{\infty}F(u, v)e^{i2\pi (ux+vy)}dudv.
\end{equation}
Here, $f(x, y)$ represents the signal of input, $ F( u, v ) $is the frequency signal after Fourier transform, $ u, v $ is the frequency spectrum, and $f(x, y)\leftrightarrow F(u, v)$ represents the transformation and inverse transformation between the signal and the frequency.

Similar with the back-propagation derivation of one-dimension convolution in frequency domain, the back-propagation of two-dimensional convolution in frequency domain can be derivative as the following theorem.
\begin{thm}\label{thm:2D_conv_derivate}
	The derivative of convolution in the spatial domain is equivalent to doing the operation of $ i2\pi F( u, v) $ in the frequency domain, namely $ \frac{\partial f(x, y)}{\partial x}\leftrightarrow i2\pi F(u, v) $.
\end{thm}
The proof of Theorem~\ref{thm:2D_conv_derivate} is in the supplemental material (Please see Appendix A.3).

Based on Theorem~\ref{thm:2D_conv_derivate}, the 2D convolution of thorough frequency transform can be derived. However, it should be noted that the calculation of mapping requires that the size of the convolution kernel should be consistent with the size of the input feature maps. In order to carry out the convolution point-wise product operation in the frequency domain, to the extent of avoiding excessive memory consumption, we need to reasonably design zero-padding or interpolation filling to reduce the occupation.

\subsection{Activation}

The activation function is introduced to increase the non-linearity of the neural network model. Currently, there are many widely used non-linear activation functions, e.g. ReLU, Sigmoid, tanh, etc. Here, we mainly deduce the ReLU and Sigmoid activation functions in the frequency domain. 

\paragraph{Sigmoid}
The Sigmoid function maps a real number to the interval (0,1), and is frequently used for binary classification. As a non-linear activation function, it is smooth, strictly monotonous, and easy to differentiate. However, due to the large amount of computation, the gradient may easily disappear during backpropagation. Formally, the Sigmoid function has the following form:
\begin{equation}\label{eqn:Sigmoid}
S(x)=\frac{1}{1+e^{-x}}=\frac{e^{x}}{e^{x}+1}.
\end{equation}

Mathematically, it satisfies the absolutely integrable requirement $\int_{-\infty }^{\infty }| f( t ) |< \infty$, and piece-wise smooth. Here, we derive the Fourier transform of Sigmoid function by giving the following Theorem.
\begin{thm}
	The Fourier transform of the Sigmoid activation function is $ S(x) \leftrightarrow \frac{ie^{x-i \omega x}}{ \omega +i}\cdot _{2}F_{1}( 1,1-i \omega; 2-i \omega;-e^{x}) $, and its derivative exists.
\end{thm}
\begin{proof}
	\begin{align}
	&F( \omega )=\int_{-\infty }^{\infty }S( x )e^{-i \omega x}dx= \int_{-\infty }^{\infty }\frac{e^{(1 -i \omega )x}}{e^{x}+1}dx
	\\&=\frac{ie^{x-i \omega x}}{ \omega +i}( 1+ \frac{1\cdot ( 1-i \omega )}{1!( 2-i \omega )}\cdot ( -e^{x} )^{1}+\frac{1\cdot ( 1+1 )( 1-i \omega )( 1-i \omega +1 )}{2!( 2-i \omega )( 2-i \omega +1 )}\cdot ( -e^{x} )^{2}+ \cdots)
	\\&=\frac{ie^{x-i \omega x}}{ \omega +i}\cdot _{2}F_{1}( 1,1-i \omega; 2-i \omega;-e^{x}),
	\end{align}
	in which spectral range $\omega \in ( -\infty ,\infty )$, and
	$ _{2}F_{1}( a,b; c;z ) $ is hypergeometric function, specifically 
	\begin{equation}\label{eqn:hyper_func}
	_{2}F_{1}(a,b;c;z)=1+\frac{ab}{1!c}z+\frac{a(a+1)b(b+1)}{2!c(c+1)}z^{2}+\cdots=\sum_{n=0}^{\infty}\frac{(a)_{n}(b)_{n}}{(c)_{n}}\frac{z^{n}}{n!}.
	\end{equation} 
	Its derivative is as follows:
	\begin{equation}
	\frac{\partial }{\partial x}( \frac{e^{-i \omega x}}{e^{-x}+1} )= \frac{e^{x-i \omega x}( 1-i \omega ( e^{x}+1 ) )}{( e^{x}+1 )^{2}}.	
	\end{equation}
\end{proof}

\paragraph{ReLU}
Different from the non-linear Sigmoid function, ReLU is linear which makes it less vulnerable to the gradient explosion problem. Currently, ReLU is widely used as an activation function in deep learning. Typically, the definition of ReLU activation function can be written
as follows:
\begin{equation} \label{eqn:ReLU}
f(x)=
\begin{cases}
0, & \text{ if } x< 0, \\
x, & \text{ if } x\geq 0.
\end{cases}
\end{equation}
Suppose that $x$ belongs to $(0,k)$, then the ReLU activation function will satisfy the absolutely integrable and piece-wise smooth requirements:
$ \int_{-\infty }^{\infty }| f( x ) |< \infty $.
Accordingly, we give the following Theorem.
\begin{thm}
	The Fourier transform of ReLU activation function is related to the independent variables and spectrum, namely $ f(x) \leftrightarrow \frac{e^{-i \omega k}(1+i\omega k)-1}{ \omega^2} $, and its derivative exists.
\end{thm}
\begin{proof}
	\begin{align}
	&F( \omega )=\int_{-\infty }^{\infty }f( x )e^{-i \omega x}dx=\int xe^{-i \omega x}dx=\frac{e^{-i \omega x}( 1+i \omega x )}{ \omega^{2}}
	\\&= \int_{0}^{k }x\cdot e^{-i \omega x}dx= -\frac{1 }{i \omega}\int_{0}^{k}x\frac{\mathrm{d}e^{-i \omega x} }{\mathrm{d} x}dx
	\\&=\frac{e^{-i \omega k}(1+i\omega k)-1}{ \omega^2}, 
	\end{align}
	
where the spectrum $\omega \in ( -\infty ,\infty )$. Also, the back-propagation frequency ReLU has the following form:
	\begin{equation}
	\frac{\partial}{\partial x}(f(x)e^{-i \omega x})=\frac{\partial}{\partial x}(x \cdot e^{-i \omega x} )=e^{-i \omega x}(1-i\omega x).	
	\end{equation}
\end{proof}
We also provide another way of thinking, namely ReLU can also be written as follows:
\begin{equation}
R(x)=f(x)u(x),	
\end{equation}
in which $ f(x)=x $, $ u(x) $ is the Heaviside step function $
u(x)=
\begin{cases}
0,& t< 0,\\
1,& t\geq 0.
\end{cases}
$. Due to the 
$ \int_{-\infty}^{\infty}|u(x)|=\int_{0}^{\infty}dx\rightarrow \infty $, it's not meet the absolutely integrable, so we can't take the Fourier transform directly. However, we can solve this problem by resorting to the $ \delta $ function and give the following Theorem.

\begin{thm} \label{thm:ReLU}
	The ReLU activation function in the frequency domain is a special $ \delta $ function, namely $ R(x)\leftrightarrow C(\pi \delta ( \omega )+\frac{1}{i \omega }) $, where C means a constant or function.
\end{thm}
\begin{proof}
	Considering that $ f(x)=1 $ is the functional limit of  $f(x)=e^{-\beta |x|} $ when $ \beta \rightarrow 0^{+} $,
	then the equivalent Heaviside step function has the following functional limit when $ \beta \rightarrow 0^{+} $:
	\begin{equation}
	f(x)=
	\begin{cases}
	0, & \text{ if } x< 0, \\
	e^{-\beta x}, & \text{ if } x\geq 0.
	\end{cases}
	\end{equation}
	
	Its Fourier transform is (using the known limit and integral formulas of complex functions):
	\begin{equation}
	F( \omega )=\int_{-\infty}^{\infty}f(x)e^{-i \omega x}dx=\int_{0}^{\infty}e^{-(\beta + i \omega  )x}dx=\frac{1}{\beta + i \omega},
	\end{equation}
	
	namely $ \mcF \left \{ u( x ) \right \}=\lim_{\beta \rightarrow 0^{+}}\frac{1}{\beta +i \omega} $.
	Since the limit of a complex function is divided into its real and imaginary parts,
	the following transformation can be made:
	\begin{equation}
	\lim_{\beta \rightarrow 0^{+}}\frac{1}{\beta +i \omega}=\lim_{\beta \rightarrow 0}\frac{\beta }{\beta^{2} + \omega ^{2}}-i \lim_{\beta \rightarrow 0}\frac{ \omega }{\beta^{2}+ \omega ^{2}}.
	\end{equation}
	Based on the limit of a Lorentz linear function that $ \lim_{\beta \rightarrow 0}\frac{\beta }{\beta^{2} + \omega ^{2}}=\pi \delta(\omega) $, we can further obtain the Fourier
	transform of the Heaviside step function $ u( x ) $ as follows:
	\begin{equation}
	F( \omega )=\lim_{\beta \rightarrow 0^{+}}\frac{1}{\beta +i \omega}=\pi \delta ( \omega )+\frac{1}{i \omega }.
	\end{equation}
\end{proof}
$ \delta(\cdot)$ is the impulse function and has a wide range of applications in quantum mechanics, classical physics, and many other disciplines. It has the following properties:
$ \delta ( x-x_{0} )=
\begin{cases}
0, & \text{ if } x\neq x_{0}, \\
\infty ,& \text{ if } x= x_{0}.
\end{cases} $, and $ \int_{-\infty}^{\infty} \delta(x-x_{0})dx=1$.
So Theorem \ref{thm:ReLU} is proved.

\subsection{Pooling}

Pooling basically boils down to sub-sampling with no parameters to learn. Compared with average pooling, maximum pooling is less effective in retaining location information, that is, it does not save enough information and only reflects very local information. Therefore, the appropriate pooling strategy is selected according to the problem to be solved. Spectral Pooling uses discrete Fourier transform(DFT) to truncate the representation in the frequency domain to achieve the pooling function in the frequency domain. Max pooling can be approximately defined as coherency function convoluted with the spectral feature in the frequency domain. Here we give the theoretical derivation of average pooling in the frequency domain.

\begin{thm}\label{thm:pooling}
	Pooling in the frequency can be defined sinc function.
\end{thm}
\begin{proof}
	Average pooling is approximately defined as the following expression:
	\begin{equation}
	h(x, y)=
	\begin{cases}
	\frac{1}{WH}, & \text{ if }| x |< \frac{W}{2}, | y |< \frac{H}{2}, \\
	0, & \text{ else }.
	\end{cases}
	\end{equation}
	Its Fourier transform is as follows:
	\begin{equation}
	\begin{split}
	&\mcF \left \{ h \right \}=F( u, v )=\int_{-\infty}^{\infty}\int_{-\infty}^{\infty}h( x, y )e^{-j2\pi ( ux+vy )}dxdy
	\\&=\int_{-\frac{W}{2}}^{\frac{W}{2}}\frac{1}{W}e^{-j2\pi ux}dx\int_{-\frac{H}{2}}^{\frac{H}{2}}\frac{1}{H}e^{-j2\pi vy}dy
	\\&=\frac{sin( \pi Wu )}{\pi Wu}\cdot \frac{sin( \pi Hv )}{\pi Hv}=sinc( Wu )\cdot sinc( Hv ).
	\end{split}
	\end{equation}
\end{proof}
According to Theorem~\ref{thm:pooling}, the Fourier transform and derivatives of average pooling exist, and the derivatives are related to sines and cosines.  
On the frequency domain, this operation corresponds to the element-wise multiplication of the spectrum denoted as $sinc( Wu )$ and $sinc( Hv )$, which is efficient and low-computational cost. Generally, the operation of pooling is equivalent to a low-pass filtering process.



\subsection{Fully connected function}

The previous layers of the fully connected (FC) layer feed the extracted features into it for classification. FC can be understood as convolution operation, and its function is to map feature representation to a value or others, which can reduce the influence of feature location on classification results. However, as it ignores spatial structure characteristics, FC is not suitable for tasks requiring location information, such as segmentation. Moreover, the number of parameters in FC layers is large, which increases the burden of training. So the backbone of ResNet\cite{ref23} adopt Global Average Pooling (GAP) to replace FC for the fusion after feature extraction. At present, the existing research work of li et al. \cite{ref21} has proved that GAP is a special column of the feature composition in the frequency domain, so it is feasible to realize the operation of FC function in the frequency domain.

\subsection{Cross entropy loss}

For binary classification, the cross entropy loss can be written as $ -(ylogp+(1-y)log(1-p)) $. For the given label $y$ and the probability $p$, we can get the follow derivation in the frequency domain.
\begin{align}
&F(\omega)=\int_{-\infty}^{\infty}-(ylogp+(1-y)log(1-p))e^{-i\omega x}dx
\\&=-y\int_{-\infty}^{\infty}logpe^{-i\omega x}dx+(y-1)\int_{-\infty}^{\infty}log(1-p)e^{-i\omega x}dx
\\&=-y\mcF \left \{ logp \right \}+(y-1)\mcF \left \{ log(1-p) \right \}.
\end{align}
Furthermore, we transform it into an exponential function, 
\begin{align}
&e^{-(ylogp+(1-y)log(1-p))} = e^{-ylogp}e^{(y+1)log(1-p)} 
\\&= p^{-y}(1-p)^{(y-1)},
\end{align}
then derive the following another form in frequency domain
\begin{align}
&F(\omega)=\int_{-\infty}^{\infty}p^{-y}(1-p)^{(y-1)}e^{-i\omega x}dx
\\&= \frac{p^{-y}(1-p)^{(y-1)}e^{-i\omega x}}{-i\omega}.
\end{align}

\section{From Fourier to Laplace transform}
\label{LT}
Laplace transform is introduced on the basis of Fourier transform. Because Laplace introduced complex field and complex variable function, its integral transformation kernel kind of real exponential function factor. After taking the Laplace transform from the spatial domain to the complex frequency domain, the discrete Laplace transform is called the Z transform.

For any function $ g(t)( t \geq 0 ) $, in order to make its Fourier transform exist in $ (-\infty ,\infty ) $, we first multiply it with the Heaviside step function $ u( t ) $. In order to satisfy the absolute integrability condition, we further multiply it with the decay factor of $ exp( -\beta t )( \beta > 0 ) $, then take the FT of $ g( t )u( t )exp( -\beta t ) $ as follows:
\begin{equation}
\int_{-\infty }^{\infty }g( t )u( t )exp( -\beta t )e^{-i \omega t}dt=\int_{0}^{\infty }f( t )e^{-pt}dt,
\end{equation}
in which
$
u(t)=
\begin{cases}
0,& t< 0,\\
1,& t\geq 0.
\end{cases}
$,
$ p=\beta +i \omega $, $ f( t )=g( t )u( t ) $. Then, we get the following definition of the Laplace transform.
\begin{defn}\label{defn:LP}
	Given $f(t)(t> 0)$, its general form of Laplace transform is as follows:
	\begin{equation}\label{eqn:LP}
	F( p )=\int_{0}^{\infty }f( t )e^{-pt}dt,
	\end{equation}
	in which the parameter p is a complex number and the real part is positive. In
	practical applications, we usually take p as a positive real number.
\end{defn}	
Accordingly, the LT of $ f( t ) $ can be written as 
$ F( p )=\mcL \left \{ f( t ) \right \} $,
and the inverse LT is $ f( t )=\mcL^{-1} \left \{ F( p ) \right \} $, generally written as follows:
\begin{equation}
F( p )\leftrightarrow f( t ).
\end{equation}
The sufficient and unnecessary conditions for the existence of the Laplace transform are:
the function of $ f( t ) $ is piece-wise continuous in the interval of
$ [ 0,\infty  ) $, $ f(t)=0 $ when $ t< 0 $. That is to say the integral of the convolution can be reduced as follows.
\begin{defn}\label{dfn:conv_LT}
	The convolution of LT is defined by the Laplace transform and the definition of the
	Fourier convolution: $f_{1}( t )\ast f_{2}( t )=\int_{0}^{t}f_{1}( \tau )f_{2}( t-\tau )d\tau$.
\end{defn}
Here, * means convolution. Based on Definition \ref{defn:LP} and \ref{dfn:conv_LT}, the correspondence relationship between Fourier transform and  Laplace transform can be derived, based on which we give the following theorem.
\begin{thm}\label{L_conv}
	The spatial convolution is equivalent to the point-wise product in the frequency domain: $ F_{1}(p)F_{2}(p)\leftrightarrow f_{1}(t)\ast f_{2}(t) $.
\end{thm}
%

The proof of Theorem~\ref{L_conv} is in the supplemental material (Please see Appendix A.4).We give the derivation of the Laplace Transform of the Sigmoid activation function in this section. As the Sigmoid function satisfies absolutely integrable requirement
$ \int_{-\infty }^{\infty }| f( t ) |< \infty $, and piece-wise smooth, so it's definitely satisfies the requirements of the Laplace transform. Here, we give the  Laplace transform of Sigmoid through the following Theorem.
\begin{thm}
	The Laplace transform of the Sigmoid activation function is $ S(x) \leftrightarrow \frac{e^{x-px}}{1-p}\cdot _{2}F_{1}( 1,1-p; 2-p;-e^{x} ) $, and its derivative exists.
\end{thm}
\begin{proof}
	\begin{align}
	&\mcL\left \{ S(x) \right \} =\int_{-\infty }^{\infty }S( x )e^{-px}dx= \int_{-\infty }^{\infty }\frac{e^{(1 -px )x}}{e^{x}+1}dx
	\\&=\frac{e^{x-px}}{1-p}( 1+ \frac{1\cdot ( 1-p )}{1!( 2-p )}\cdot ( -e^{x} )^{1}+\frac{1\cdot ( 1+1 )( 1-p )( 1-p+1 )}{2!( 2-p )( 2-p+1 )}\cdot ( -e^{x} )^{2}+ \cdots)
	\\&=\frac{e^{x-px}}{1-p}\cdot _{2}F_{1}( 1,1-p; 2-p;-e^{x} ), 
	\end{align}
	in which spectral spectrum $\omega \in ( -\infty ,\infty )$, $p$ is an exponential factor without complex number, and 
	$ _{2}F_{1}(a,b; c;z) $ is the hyper-geometric function defined in Eqn~(\ref{eqn:hyper_func}). 
	Its derivative is as follows:
	\begin{equation}
	\frac{\partial }{\partial x}( \frac{e^{-px}}{e^{-x}+1} )= \frac{e^{x-px}( pe^{x} +p-1 )}{( e^{x}+1 )^{2}}.	
	\end{equation}
\end{proof}
The operations of the Laplace transform are similar to the Fourier transform, and we get the following theorem about ReLU of frequency domain.
\begin{thm}\label{L_ReLU}
	The Laplace transform of ReLU activation function is related to the independent variables and spectrum, namely $ f(x) \leftrightarrow -\frac{e^{-pk}(1+pk)+1}{ p^2} $, and its derivative exists.
\end{thm}
\begin{proof}
Suppose we limit the upper limit of $ x $ belong to $ (0,k) $ similar to the Fourier transform of ReLU, the ReLU activation function $f(x) $ Eqn.~(\ref{eqn:ReLU}) will satisfy the absolutely integrable and piece-wise smooth requirements.

	\begin{align}
	&\mcL\left \{ f(x) \right \}=\int_{-\infty }^{\infty }f( x )e^{-px}dx=\int xe^{-px}dx=-\frac{e^{-px}( 1+px )}{ p^{2}}
	\\&= \int_{0}^{k }x\cdot e^{-px}dx= -\frac{1 }{p}\int_{0}^{k}x\frac{\mathrm{d}e^{-px} }{\mathrm{d} x}dx
	\\&=-\frac{e^{-pk}(1+pk)+1}{ p^2}, 
	\end{align}	
	with the spectrum spectrum $\omega \in ( -\infty ,\infty )$. Its derivative is as follows:
\begin{equation}
\frac{\partial}{\partial x}(f(x)e^{-px})=\frac{\partial}{\partial x}(x \cdot e^{-px} )=e^{-px}(1-px),	
\end{equation}

\begin{equation}
\frac{\partial}{\partial p} (-\frac{ e^{-p x}(1+px)+1}{p ^{2}})=\frac{e^{-px}(p^{2}x^{2}+2px+2e^{pk}+2)}{p^{3}}.
\end{equation}
\end{proof}
In the derivation of the other Laplace transform operations, we can compute CNNs more efficiently due to the exponential factors circumvent complex numbers.

\section{Discussion and prospect}
In this paper, we studied the frequency operation of convolution neural networks theoretically to facilitate the construction of end-to-end frequency neural networks. Based on the Fourier transform theory, feed-forward and back-propagation frequency operations of basic network components are derived including convolutional computation, average pooling, Sigmoid, ReLU, and fully connected computation. To avoid the problem raised by the complex number space in the Fourier Transform, we further extend the frequency network from Fourier to Laplacian where the Laplace frequency operations can be completely performed in the real space. Based on all the above, by providing theoretical derivation, we built the deeper theoretical fundament for the complete frequency convolution neural network in both the complex or real number domain.Admittedly, this work still has some defects, such as the lack of corresponding numerical experiments, and the challenge of low practicability in the case of small convolution kernel size. 

Our work will promote the work of frequency domain convolutional neural network, as the starting point of relevant research. Although we admit that this work is not perfect, it may be still promising in reducing the operation time and computational complexity with sufficient theoretical support. We expect that the theoretical extension and discussion about frequency CNN may essentially benefit the acceleration of the training and inference of CNN, and broaden the development prospect of neural networks in the frequency domain.

\newpage
\appendix

\section{Appendix}

\subsection{Proof of Theorem~\ref{thm:conv}}
Based Definition~\ref{dfn:1d_conv} and Eqn.~\ref{eqn:Fconv},we provide the proof of Theorem~\ref{thm:conv}, i.e. the correspondence between convolution in spatial convolution and the point-wise product in the frequency domain.

\begin{proof}
	\begin{align}
	&\int_{-\infty }^{\infty }f_{1}( x )\ast f_{2}( x )e^{-i \omega x}dx=\int_{-\infty }^{\infty }[ \int_{-\infty }^{\infty } f_{1}( \xi  )f_{2}( x-\xi )d\xi ]e^{-i \omega x}dx
	\\&=\int_{-\infty }^{\infty }f_{1}( \xi )[ \int_{-\infty }^{\infty } f_{2}( x-\xi )e^{-i \omega (x-\xi)}dx ]e^{-i \omega \xi }d\xi\indent ( y=x-\xi )
	\\&=\int_{-\infty }^{\infty }f_{1}( \xi )[ \int_{-\infty }^{\infty } f_{2}( y )e^{-i \omega y}dy ]e^{-i \omega \xi }d\xi
	\\&=\int_{-\infty }^{\infty }f_{1}( \xi ) F_{2}( \omega )e^{-i \omega  \xi }d\xi=F_{2}( \omega )\int_{-\infty }^{\infty }f_{1}( \xi )e^{-i \omega \xi }d\xi
	\\&=F_{1}( \omega )F_{2}( \omega ).
	\end{align}
\end{proof}

\paragraph{B. Sigmoid of FT}
The Sigmoid function is known to have the following form:
\begin{equation}
S(x)=\frac{1}{1+e^{-x}}=\frac{e^{x}}{e^{x}+1}
\end{equation}
Satisfy absolutely integrable requirement
$
\int_{-\infty }^{\infty }| f( t ) |< \infty
$, and piece-wise smooth. So the Fourier transform of the Sigmoid activation
function is
\begin{proof}
	\begin{align}
	&F( \omega )=\int_{-\infty }^{\infty }S( x )e^{-i \omega x}dx= \int_{-\infty }^{\infty }\frac{e^{(1 -i \omega )x}}{e^{x}+1}dx
	\\&=\frac{ie^{x-i \omega x}}{ \omega +i}( 1+ \frac{1\cdot ( 1-i \omega )}{1!( 2-i \omega )}\cdot ( -e^{x} )^{1}+\frac{1\cdot ( 1+1 )( 1-i \omega )( 1-i \omega +1 )}{2!( 2-i \omega )( 2-i \omega +1 )}\cdot ( -e^{x} )^{2}+ \cdots)
	\\&=\frac{ie^{x-i \omega x}}{ \omega +i}\cdot _{2}F_{1}( 1,1-i \omega; 2-i \omega;-e^{x})
	\end{align}
\end{proof}
in which spectral range $\omega \in ( -\infty ,\infty )$, and
$ _{2}F_{1}( a,b; c;z ) $ is hypergeometric function. It's derivative as follows:
\begin{equation}
	\frac{\partial }{\partial x}( \frac{e^{-i \omega x}}{e^{-x}+1} )= \frac{e^{x-i \omega x}( 1-i \omega ( e^{x}+1 ) )}{( e^{x}+1 )^{2}}	
\end{equation}

\subsection{Proof of theorem about the convolutional derivation}\label{Proof_T3.1}
	The derivative of convolution in the spatial domain is equivalent to doing the operation of $ i \omega F( \omega ) $ in the frequency domain, namely $ \frac{df(x)}{dx}\leftrightarrow i \omega F( \omega ) $.
According to the above theorem, $f(x)'\leftrightarrow i\omega F(\omega)$, in which $f(x)'$ means the derivative of $f(x)$. 
\begin{proof}
	\begin{align}
	&\frac{\mathrm{d} f(x)}{\mathrm{d} x}\leftrightarrow\int_{-\infty}^{\infty}\frac{\mathrm{d} f( x )}{\mathrm{d} x}e^{-i \omega x}dx=\int_{-\infty}^{\infty}e^{-i \omega x}df( x )
	\\&=[ f( x )e^{-i \omega x}]_{-\infty}^{\infty}+i \omega \int_{-\infty }^{\infty }f(x)e^{-i \omega x}dx= i \omega F( \omega ). 
	\end{align}
\end{proof}

\subsection{Proof of Theorem~\ref{thm:2D_conv_derivate}.}\label{Proof_T3.3}
\begin{proof}
	\begin{align}
	&\frac{\partial f(x, y)}{\partial x}\leftrightarrow \int_{-\infty}^{\infty}\int_{-\infty}^{\infty}\frac{\mathrm{d} f(x, y)}{\mathrm{d} x}e^{-i2\pi (ux+vy)}dxdy=\int_{-\infty}^{\infty}[\int_{-\infty}^{\infty}e^{-i2\pi (ux+vy)}df(x, y)]dy
	\\&= \int_{-\infty}^{\infty}[[f(x, y)e^{-i2\pi (ux+vy)}]_{-\infty}^{\infty}+i2\pi\int_{-\infty}^{\infty}f(x, y)e^{-i2\pi (ux+vy)}dx]dy
	\\&=0+i2\pi\int_{-\infty}^{\infty}\int_{-\infty}^{\infty}f(x, y)e^{-i2\pi (ux+vy)}dxdy=i2\pi F(u, v).
	\end{align}
\end{proof}

\subsection{Proof of Theorem~\ref{L_conv}}\label{Proof_Lap}
\begin{proof}
	\begin{align}
	&F_{1}( p )F_{2}( p )=[ \int_{0}^{\infty }f_{1}( u )e^{-pu}du ][ \int_{0}^{\infty }f_{2}( v )e^{-pv}dv ]
	\\&=\int_{t=0}^{\infty }\int_{u=0}^{t}e^{-pt}f_{1}( u )f_{2}( t-u )dudt
	\\&=\int_{t=0}^{\infty }e^{-pt}[ \int_{u=0}^{t}f_{1}( u )f_{2}( t-u )du ]dt.
	\end{align}
	
	Based on the Laplacian convolution definition, we get the following result:
	\begin{align}
	&F_{1}(p)F_{2}(p)=\int_{t=0}^{\infty }e^{-pt} [ f_{1}( t )\ast f_{2}( t ) ]dt
	\\&=\mcL\left \{ f_{1}(t)\ast f_{2}(t) \right \}, \nonumber
	\end{align}
\end{proof}
in which $ \mcL $ means the Laplace transform. Hence, the convolution of primitive functions turns out to be the point-wise product of a function in the frequency domain.

%
%
%
%


\begin{thebibliography}
\medskip
{
\small
\bibitem{ref1} LeCun, Y., Bottou, L., Bengio, Y.,\ \& Haffner, P. (1998).
Gradient-based learning applied to document recognition.
{\it Proceedings of the IEEE}, 86(11), 2278-2324.

\bibitem{ref2} LeCun, Y., Kavukcuoglu, K.,\ \& Farabet, C. (2010, May).
Convolutional networks and applications in vision.
{\it In Proceedings of 2010 IEEE international symposium on circuits and systems} (pp. 253-256). IEEE.

\bibitem{ref3} Bracewell, R. N.,\ \& Bracewell, R. N. (1986).
{\it The Fourier transform and its applications}
(Vol. 31999, pp. 267-272). New York: McGraw-Hill.

\bibitem{ref4} Schiff, J. L. (1999).
{\it The Laplace transform: theory and applications}.
Springer Science \& Business Media.

\bibitem{ref5} Parkhi, O. M., Vedaldi, A.,\ \& Zisserman, A. (2015).
{\it Deep face recognition}.

\bibitem{ref6} Mathieu, M., Henaff, M.,\ \& LeCun, Y. (2013).
Fast training of convolutional networks through ffts.
{\it arXiv preprint arXiv:1312.5851}.

\bibitem{ref7} Rippel, O., Snoek, J.,\ \& Adams, R. P. (2015).
Spectral representations for convolutional neural networks.
{\it arXiv preprint arXiv:1506.03767}.

\bibitem{ref8} Vasilache, N., Johnson, J., Mathieu, M., Chintala, S., Piantino, S.,\ \& LeCun, Y. (2014).
Fast convolutional nets with fbfft: A GPU performance evaluation.
{\it arXiv preprint arXiv:1412.7580}.

\bibitem{ref9} Xu, Y.,\ \& Nakayama, H. (2019).
Shifted Spatial-Spectral Convolution for Deep Neural Networks.
{\it In Proceedings of the ACM Multimedia Asia} (pp. 1-6).

\bibitem{ref10} Pratt, H., Williams, B., Coenen, F.,\ \& Zheng, Y. (2017, September).
Fcnn: Fourier convolutional neural networks.
{\it In Joint European Conference on Machine Learning and Knowledge Discovery in Databases} (pp. 786-798). Springer, Cham.

\bibitem{ref11} Xu, Z. Q. J.,\ \& Zhou, H. (2020).
Deep frequency principle towards understanding why deeper learning is faster.
{\it arXiv preprint arXiv:2007.14313}.

\bibitem{ref12} Xu, Z. Q. J., Zhang, Y.,\ \& Xiao, Y. (2019, December).
Training behavior of deep neural network in frequency domain.
{\it In International Conference on Neural Information Processing} (pp. 264-274). Springer, Cham.

\bibitem{ref13} Gueguen, L., Sergeev, A., Kadlec, B., Liu, R.,\ \& Yosinski, J. (2018).
Faster neural networks straight from jpeg.
{\it Advances in Neural Information Processing Systems}, 31, 3933-3944.

\bibitem{ref14} Liu, Z., Xu, J., Peng, X.,\ \& Xiong, R. (2018, December).
Frequency-domain dynamic pruning for convolutional neural networks.
{\it In Proceedings of the 32nd International Conference on Neural Information Processing Systems} (pp. 1051-1061).

\bibitem{ref15} Wang, Y., Xu, C., You, S., Tao, D.,\ \& Xu, C. (2016, December).
CNNpack: Packing Convolutional Neural Networks in the Frequency Domain.
{\it In NIPS} (Vol. 1, p. 3).

\bibitem{ref16} Chen, W., Wilson, J., Tyree, S., Weinberger, K. Q.,\ \& Chen, Y. (2016, August).
Compressing convolutional neural networks in the frequency domain.
{\it In Proceedings of the 22nd ACM SIGKDD International Conference on Knowledge Discovery and Data Mining} (pp. 1475-1484).

\bibitem{ref17} Xu, K., Qin, M., Sun, F., Wang, Y., Chen, Y. K.,\ \& Ren, F. (2020).
Learning in the frequency domain.
{\it In Proceedings of the IEEE/CVF Conference on Computer Vision and Pattern Recognition} (pp. 1740-1749).

\bibitem{ref18} Lee-Thorp, J., Ainslie, J., Eckstein, I.,\ \& Ontanon, S. (2021).
FNet: Mixing Tokens with Fourier Transforms.
{\it arXiv preprint arXiv:2105.03824}.

\bibitem{ref19} Li, Z., Kovachki, N., Azizzadenesheli, K., Liu, B., Bhattacharya, K., Stuart, A.,\ \& Anandkumar, A. (2020).
Fourier neural operator for parametric partial differential equations.
{\it arXiv preprint arXiv:2010.08895}.

\bibitem{ref20} Qiao, C., Li, D., Guo, Y., Liu, C., Jiang, T., Dai, Q.,\ \& Li, D. (2021).
Evaluation and development of deep neural networks for image super-resolution in optical microscopy.
{\it Nature Methods}, 18(2), 194-202.

\bibitem{ref21} Qin, Z., Zhang, P., Wu, F.,\ \& Li, X. (2020).
FcaNet: Frequency Channel Attention Networks.
{\it arXiv preprint arXiv:2012.11879}.

\bibitem{ref22} Ko, J. H., Mudassar, B., Na, T.,\ \& Mukhopadhyay, S. (2017, June).
Design of an energy-efficient accelerator for training of convolutional neural networks using frequency-domain computation.
{\it In 2017 54th ACM/EDAC/IEEE Design Automation Conference (DAC)} (pp. 1-6). IEEE.

\bibitem{ref23} He, K., Zhang, X., Ren, S.,\ \& Sun, J. (2016, October).
Identity mappings in deep residual networks.
{\it In European conference on computer vision} (pp. 630-645). Springer, Cham.

\bibitem{ref24} Smith, J. S.,\ \& Wilamowski, B. M. (2018, June).
Discrete cosine transform spectral pooling layers for convolutional neural networks.
{\it In International Conference on Artificial Intelligence and Soft Computing} (pp. 235-246). Springer, Cham.

\bibitem{ref25} He, K.,\ \& Sun, J. (2015). 
Convolutional neural networks at constrained time cost. 
{\it In Proceedings of the IEEE conference on computer vision and pattern recognition} (pp. 5353-5360). 

\bibitem{ref26} Krizhevsky, A. (2014). 
One weird trick for parallelizing convolutional neural networks. 
{\it arXiv preprint arXiv:1404.5997.}

\bibitem{ref27} Krizhevsky, A., Sutskever, I.,\ \& Hinton, G. E. (2012). 
Imagenet classification with deep convolutional neural networks. 
{\it Advances in neural information processing systems}, 25, 1097-1105.
}
\end{thebibliography}
\end{document}